\documentclass[aps,letterpaper,nofootinbib,showpacs,preprintnumbers,amsmath,amssymb]{revtex4}
\usepackage{latexsym}
\usepackage{hyperref}
\usepackage{amssymb}
\usepackage{graphicx}
\usepackage{amsmath}
\usepackage{amsthm,amssymb,amstext,amscd,amsfonts,amsbsy,amsxtra,latexsym}
\usepackage{wrapfig}
\usepackage{color}
\usepackage{pdflscape}
\usepackage{lscape}
\usepackage{tabularx}
\usepackage{lipsum}
\usepackage{rotating}
\usepackage{pifont}
\usepackage{bbding}
\usepackage{upgreek}
\usepackage{mathrsfs}
\usepackage{wasysym}
\usepackage{verbatim}

\newcommand{\lb}{\label}
\newcommand{\be}{\begin{equation}}
\newcommand{\ee}{\end{equation}}
\newcommand{\ben}{\begin{eqnarray*}}
\newcommand{\een}{\end{eqnarray*}}
\newcommand{\bea}{\begin{eqnarray}}
\newcommand{\eea}{\end{eqnarray}}
\newcommand{\md}{{\mathrm{d}}}
\newcommand{\beq}{\begin{equation}}
\newcommand{\eeq}{\end{equation}}
\newcommand{\beqn}{\begin{equation}\nonumber}
\newcommand{\bean}{\begin{eqnarray}\nonumber}
\DeclareMathAlphabet{\mathpzc}{OT1}{pzc}{m}{it}

\def\cR{{\mathcal R}}
\def\cH{{\mathcal H}}

\def\cQ{{\mathcal Q}}

\def\cD{{\mathcal D}}
\def\cL{{\mathcal L}}
\def\cM{{\mathcal M}}

\def\cC{\underline{{\mathcal C}}}

\def\bbN{{\mathbb N}}

\def\bbR{{\mathbb R}}
\def\bbS{{\mathbb S}}
\def\dV{{\mbox{dV}^4_g}}

\newtheorem{thm}{Theorem}[section]

\newtheorem{ass}[thm]{Assumption}
\newtheorem{lem}[thm]{Lemma}
\newtheorem{prop}[thm]{Proposition}
\newtheorem{cor}[thm]{Corollary}

\newtheorem{rmk}[thm]{Remark}
\renewenvironment{proof}[1][Proof]{\begin{trivlist}
\item[\hskip \labelsep {\bfseries #1:}]}{\qed\end{trivlist}}

\newcommand{\nabb}{\mbox{$\nabla \mkern-13mu /$\,}}
\newcommand{\Dell}{\mbox{$\Delta \mkern-13mu /$\,}}
\newcommand{\gSlash}{\mbox{$g \mkern-9mu /$\,}}

\newcommand{\dVt}{d\,V^3}
\newcommand{\prp}{\partial_r\psi}
\newcommand{\Mint}{{\mathcal M}_{\text int}}

\newcounter{mnotecount}[section]

\renewcommand{\themnotecount}{\thesection.\arabic{mnotecount}}

\newcommand{\mnote}[1]
{\protect{\stepcounter{mnotecount}}$^{\mbox{\footnotesize $%
\!\!\!\!\!\!\,\bullet$\themnotecount}}$ \marginpar{
\raggedright\tiny\em $\!\!\!\!\!\!\,\bullet$\themnotecount: #1} }

\newcommand{\Srn}{\Sigma_{r_0}}
\newcommand{\Sr}{\Sigma_{r}}

\def\XXint#1#2#3{{\setbox0=\hbox{$#1{#2#3}{\int}$ }
\vcenter{\hbox{$#2#3$ }}\kern-.6\wd0}}

\def\thesection{\arabic{section}}

\makeatletter
\def\p@subsection{}
\makeatother

\begin{document}


\begin{center}

{\bf {\Large
BOUNDED ENERGY WAVES ON THE BLACK HOLE INTERIOR OF REISSNER-NORDSTR\"OM-DE SITTER}}\\

\bigskip
Jo\~ao L. Costa\footnote{e-mail address: jlca@iscte.pt}{${}^{,\dagger}$}{${}^{,\star}$} and
Anne T. Franzen\footnote{e-mail addresses: anne.franzen@tecnico.ulisboa.pt}{${}^{,\star}$}

\bigskip
{\it {${}^\dagger$}Instituto Universit\'ario de Lisboa (ISCTE-IUL), Lisboa, Portugal.}\\
\bigskip
{\it {${}^\star$}Center for Mathematical Analysis, Geometry and Dynamical Systems,}\\
{\it Instituto Superior T\'ecnico, Universidade de Lisboa,}\\
{\it Av. Rovisco Pais, 1049-001 Lisboa, Portugal}

\end{center}
\medskip

\centerline{ABSTRACT}

\noindent
Motivated by the Strong Cosmic Censorship Conjecture, in the presence of a cosmological constant,
we consider solutions of the scalar wave equation $\Box_g\phi=0$
on fixed subextremal Reissner--Nordstr\"om--de Sitter backgrounds $(\cM, g)$, without imposing symmetry assumptions on $\phi$.
We provide a sufficient condition, in terms of surface gravities and a parameter
for an exponential decaying Price law, for a local energy of the waves to remain bounded up to the Cauchy horizon.
The energy we consider controls, in particular, regular transverse derivatives at the Cauchy horizon; this  allows
us to extend the solutions with bounded energy, to the Cauchy horizon, as functions in $C^0\cap H^1_{loc}$. Our results correspond
to another manifestation of the potential breakdown
of Strong Cosmic Censorship in the positive cosmological constant setting.

\medskip

\tableofcontents


\section{Introduction}

The analysis of the geometry of the black hole interior region of some exact solutions of the Einstein equations
reveals a disturbing phenomenon that puts into question the deterministic character
of General Relativity: global uniqueness, in full generality, fails for the Einstein equations!
Nonetheless, there are epistemological and physical reasons to believe that such pathological
features are unstable and that global uniqueness holds generically.  The Strong Cosmic Censorship (SCC) Conjecture  substantiates this expectation.

The issue of under which regularity constraints, for the metric, one should formulate SCC is a subtle one (see~\cite{christo_form, m_bh, joao1, chrusciel, earman}).
Here, let us just mention that after the non--linear analysis of the spherically symmetric Einstein-Maxwell-scalar field system,
by Dafermos~\cite{m_stab, m_interior}, that followed previous fundamental breakthroughs~\cite{poisson, hiscock, amos}, the expectation became that,
in the context of black hole spacetimes,
for generic initial data the metric extends, beyond the  Maximal Globally Hyperbolic
Development (MGHD)~\footnote{Informally the MGHD is the largest Lorentzian manifold $(\cM,g)$ determined, via Einstein's equations, by the initial data.
Given an extension $(\tilde \cM,\tilde g)$ of the MGHD the boundary of $\cM$ within $\tilde\cM$ is known as the Cauchy horizon ${\cal CH}$.
See~\cite{hans_cauchy}
for precise definitions.}, in $C^0$ but not in  $C^2$.
The
understanding of this fact
leads to a formulation of SCC that forbids the generic existence
of continuous extensions of the metric with squared integrable connection coefficients. This guarantees that (generically)
no extension will solve the field equations, even in a weak sense; a proof of such formulation of the conjecture
would save a version of determinism for General Relativity.

By now, there is strong evidence~\cite{m_bh, luk3} that this form of SCC holds for large classes of asymptotically flat initial data. Nonetheless, in the
cosmological setting, i.e., if we add a positive cosmological constant to the Einstein equations, the situation is quite different as a consequence of the expected ``faster" decay of gravitational perturbations, in the exterior region (see~\cite{brady1,brady2,chambers} and references therein for the original heuristic and numerical predictions).
In fact, the series~\cite{ joao1, joao2, joao3}
provides the construction of large classes of spherically symmetric solutions of the Einstein-Maxwell-scalar field system
whose MGHD can be extended, in a highly non--unique way, with continuous metric, Christoffel symbols in $L^2_{loc}$ and scalar field $\phi\in H^1_{loc}$.

In view of the importance that a positive cosmological constant has in modern cosmology~\cite{hans_top}, the situation cannot be taken lightly and further
research is in order. This paper corresponds to a first step in trying to go beyond spherical symmetry. This will be accomplished by the study of the wave equation
\bea
\lb{waveEq}
\square_{g}\phi=0\;,
\eea
on fixed subextremal Reissner--Nordstr\"om--de Sitter backgrounds $(\cM, g)$, without imposing any symmetry assumptions on  $\phi$.

The study of the (linear) wave equation as a stepping stone for the study of the Einstein equations has a long tradition in Relativity.
In such tradition
the solutions of~\eqref{waveEq} are known as {\em scalar perturbations of the metric}.
If we interpret this ``analogy'', where the scalar field $\phi$ is at the level
of the metric $g$,  in the context of SCC the results in~\cite{joao3} translate to the following expectation: {\em On Reissner--Nordstr\"om spacetimes,
with parameters sufficiently close to extremal (but still
subextremal), solutions of the wave equation, with sufficiently ``fast" decay along the event horizon, are bounded and have bounded energy; consequently they extend,
 to the Cauchy horizon, as functions in $C^{0}\cap H_{loc}^1$.}

The main result of this paper is Theorem~\ref{mainThm} that establishes a precise realization of the previous expectation. Moreover it provides a
criterion~\eqref{mainCond} for energy boundedness,
which we expect to be sharp, in terms of the surface gravities, $\kappa_-$ and $\kappa_+$, of the Cauchy and event horizons,
and a parameter $p$ parameterizing an exponential decaying Price
law upper bound (see Assumption~\ref{pricesAss}).

Our results can be interpreted as establishing a degree of linear stability of the Cauchy horizon: in this view, the higher the
regularity of $\phi$, up to ${\cal CH}$, the more (linearly) stable the Cauchy horizon is.
Recall that a relevant potential instability of Cauchy horizons is associated to the blow up, in $L^2$, of (regular) transverse
derivatives;
this in turn is associated to {\em mass inflation}~\cite{m_stab, m_interior, joao3, jan1}.

We expect our results here to generalize to solutions of the wave equation on subextremal Kerr--de Sitter backgrounds; we plan to pursue this
goal in the near future (see~\cite{peter2} for related results concerning the small angular momentum parameter case).

\subsection{Comments on related results:}

In~\cite{anne}, one of the authors, showed uniform pointwise boundedness of waves, arising from compactly supported Cauchy data,
for the entire subextremal range of asymptotically flat Reissner--Nordstr\"om spacetime.
We expect that, by adapting the techniques there, one should be able to obtain the uniform bound~\eqref{mainBound} even in the case where~\eqref{mainCond} does not hold
~\footnote{Note that the proof in~\cite{anne}  requires the partition of the black hole interior into a larger number of subdomains;
the reason for this stems from the fact that the results in~\cite{anne} apply to the full subextremal range, while the results presented
here apply only to the subregion satisfying~\eqref{mainCond}.}.
Moreover, if~\eqref{mainCond} does not hold, one expects the $L^2$ norm of the transverse derivatives to blow up as one approaches the Cauchy
horizon (see discussion below).

In Chapter 4 of his PhD thesis~\cite{jan1} Sbierski showed, for waves with compact support along the event horizon, how to obtain
control, in $L^2$, over regular transverse derivatives, provided $2\kappa_+>\kappa_-$.
We can recover this result from ours
by taking the limit $p\rightarrow\infty$.
The elegant analysis in~\cite{jan1} was an important starting point for
our work here and, in fact, we lend some of its ideas.
%

More precisely, since Sbierski is considering the case of a compactly supported scalar field along the event horizon he is able to apply his modified
red-shift vector field~\eqref{modRS} from the start; this allows him to obtain exponential decay, with rate governed by $\kappa_+$, along a constant radius hypersurface,
$r=r_0$, in the future of the event horizon. Since here we will not be dealing with trivial data, the energy estimate associated to~\eqref{modRS},
used by Sbierski, will not be helpful in our context. To overcome this we have devised a new strategy which we will now briefly summarize.
We start by considering the standard red-shift vector field, of Dafermos and Rodnianski, to obtain, along $r=r_0$, exponential decay, with rate
governed by $\kappa<\min\{\kappa_+,2p\}$ (see Section~\ref{secRed}). This decay rate is, nonetheless, far from optimal and manifestly insufficient to obtain the desired
stability result. In order to improve it, we use an iteration scheme that works along the following lines: first, we use the energy estimates
to obtain pointwise estimates for all relevant quantities (see Section~\ref{sectionCond}); then we use the previous pointwise estimates and Sbierski's
modified red-shift vector field to obtain new energy estimates -- to capture the decay of our horizon data we need to consider co-moving regions,
whose geometry depends on a parameter $m$ (see Section~\ref{sectionImprov}); finally, we can iteratively improve our
energy estimates by a judicious choice of a sequence of parameters $m_n$ (see Section~\ref{secIt}). In the end of this process we obtain decay,
along $r=r_0$, with rate governed by  $\kappa<\min\{\kappa_+,p\}$. Finally, by a simple adaptation of Sbierski's use of the modified blue-shift
vector~\eqref{modBS}, we can propagate these estimates all the way to the Cauchy horizon and obtain the desired control over transversal derivatives,
provided $2\min\{\kappa_+,p\}>\kappa_-$ (see Section~\ref{secProp}).

Recently, in~\cite{peter2}, Hintz and Vasy obtained, by different techniques,  $H_{loc}^{1/2+\alpha/\kappa_--\epsilon}$ regularity across the Cauchy horizon,
for solutions arising from smooth Cauchy data. There, $\alpha>0$ is bounded above by the {\em spectral gap}.
Such result is clearly strongly related to our main result here, since for {\bf generic} Cauchy data, the optimal value of $p$, in Assumption~\ref{pricesAss},
and the value of $\alpha$ are expected to be the same.
But strictly speaking neither result implies the other: to see this, note for instance that in our setting even if we make $p$ arbitrarily large
(which can always be achieved but will presumably correspond to solutions which arise from non-generic Cauchy data) the condition
$2\kappa_+>\kappa_-$ is still required to obtain bounded energy waves on the black hole interior; in fact, that $p$ arbitrarily large alone
is not a sufficient condition (for energy boundedness) is corroborated by the existence of mass inflation solutions with arbitrarily
fast decaying tails along the event horizon,
 as constructed in~\cite{m_stab, m_interior, joao3}. Moreover, it goes without saying, that  different approaches reveal different
 aspects of wave propagation on the black hole interior;
concerning our approach we highlight
as advantages: the derivation of detailed pointwise and energy estimates, the clarification of the criterion~\eqref{mainCond}
and the fact that our results also apply, with minimal changes, to asymptotically flat and anti-de Sitter
Reissner--Nordstr\"om black hole interiors (see Remark~\ref{rmkRN}).

{
Still concerning stability results, for waves on the interior of extremal black holes,
 Gajic~\cite{dejan, dejan2} established bounded energy, as well as other
higher regularity statements in restricted symmetry settings. It is somewhat amusing that the instabilities of the exterior
of extremal Reissner-Nordstr\"om spacetimes~\cite{aretakis2} are accompanied by stability of their interior regions. Opposed
to this, for subextremal  Reissner-Nordstr\"om-de Sitter spacetimes it is, in some sense, the higher stability of their
exterior regions, when compared with their asymptotically flat counterparts, that is responsible for the higher stability of their Cauchy horizons.
}

Let us now briefly turn to linear instability results. The most complete statement that we are aware of is provided by the work of Luk and Oh~\cite{luk_oh},
where it is shown that, generically, linear waves are not in $H^1_{loc}$, near the Cauchy horizon of any subextremal
(asymptotically flat) Reissner--Nordstr\"om spacetime; to the best of our knowledge, this is the first work that shows Cauchy horizon instability
without having to assume some aspect of the asymptotic profile for the waves along the event horizon. Recently, Luk and Sbierski~\cite{luk_jan},
have considered the more demanding setting of the Kerr geometry and obtained blow up of transverse derivatives in $L^2$, conditional on assuming
specific upper and lower bounds for energies of the scalar field along the event horizon; see also~\cite{DafYak} for results concerning the
existence of waves, arising from scattering data,  with infinite energy on the interior of Kerr black holes.

\bigskip

Our main result here (Theorem~\ref{mainThm}) is conditional  on assuming exponential decay, in the Eddington-Finkelstein coordinates~\eqref{edd-fin-metric},
of the energy of $\phi$ along the event horizon (Assumption~\eqref{pricesAss})~\footnote{It is this ``fast" exponential decay that is instrumental in establishing our stability result and not so much the differences in the geometries of the black hole interiors, when we change the sign of $\Lambda$. In fact, we note once again, that one can easily adapt our main result, and corresponding  proof, to the case $\Lambda\leq 0$ (see Remark~\eqref{rmkRN}).}. This form of Price's law~\footnote{Strictly speaking Price's law usually refers to both upper and lower pointwise bounds. Here we will only need the upper bound in the weaker energy form of Assumption~\eqref{pricesAss}.} has been expected, for quite some time now~\cite{brady1, brady2}, to capture
the decaying profile of waves along the event horizon of cosmological black holes. Moreover, the numerology of mass inflation suggests that this
``fast'' decay, when compared to
the polynomial decay for asymptotically flat black holes, could have a strong impact on the stability of Cauchy horizons and consequently
on SCC (see~\cite{brady2, chambers} and references therein). This discussion has been recently revived by the non-linear analysis carried
out in~\cite{joao3}.

Mostly motivated by the black hole stability problem a remarkable amount of effort has been devoted to the study of the decay of solutions to the wave equation
in the exterior of cosmological black hole spacetimes~\cite{dyatlov, BonyHafner, m_SdS, volker2}.
The end product~\cite{dyatlov} of these establishes exponential decay with rate
$\alpha$ (as above). The precise relation of the spectral gap
$\alpha$ with the black hole parameters is still unclear (see for instance the discussion
in~\cite{peter2}, [Remark 2.24]).
Moreover, establishing that a lower bound also holds generically seems to remain a wide open problem in the cosmological setting~\footnote{Recently,
in the asymptotically flat  ($\Lambda=0$) case,  a considerable amount of progress has been made concerning this issue~\cite{luk_oh,luk_oh1,luk_oh2,angel}.}.


\bigskip

The recent remarkable success concerning the non-linear stability of the local region of Kerr-de Sitter, by Hintz and Vasy~\cite{peter_kerrstab,peterNewman},
and the preliminary breakthrough results of Schlue concerning the non-linear stability of the cosmological region~\cite{volkerDecay},
create the expectation that some partial resolution of SCC, at least in a neighborhood of Kerr-de Sitter, might appear in a not too distant future.
Nonetheless, as we hope our work here helps to make clear, a full understanding of SCC, even in a neighborhood of Kerr-de Sitter,
will require a very precise quantitative understanding of the decay of gravitational perturbations along the event horizon of cosmological black holes.



\section{Setup}
\lb{setup}

\subsection{The metric and ambient differential structure in Eddington-Finkelstein coordinates}
\lb{edd-fin-metric}
To set the semantic convention, whenever we refer to the Reissner-Nordstr\"om-de Sitter solution $(\cM,g)$ we mean the maximal domain
of dependence \mbox{$\cD (\Sigma)=\cM$} of a compact Cauchy hypersurface $\Sigma=\mathbb{S}^1\times\mathbb{S}^2$.
The manifold $\cM$ can be expressed by \mbox{$\cM=\cQ\times \bbS^2$}, with $\cQ$ depicted in Figure \ref{RN_ganz}.
{\begin{figure}[ht]
\centering
\includegraphics[width=0.8\textwidth]{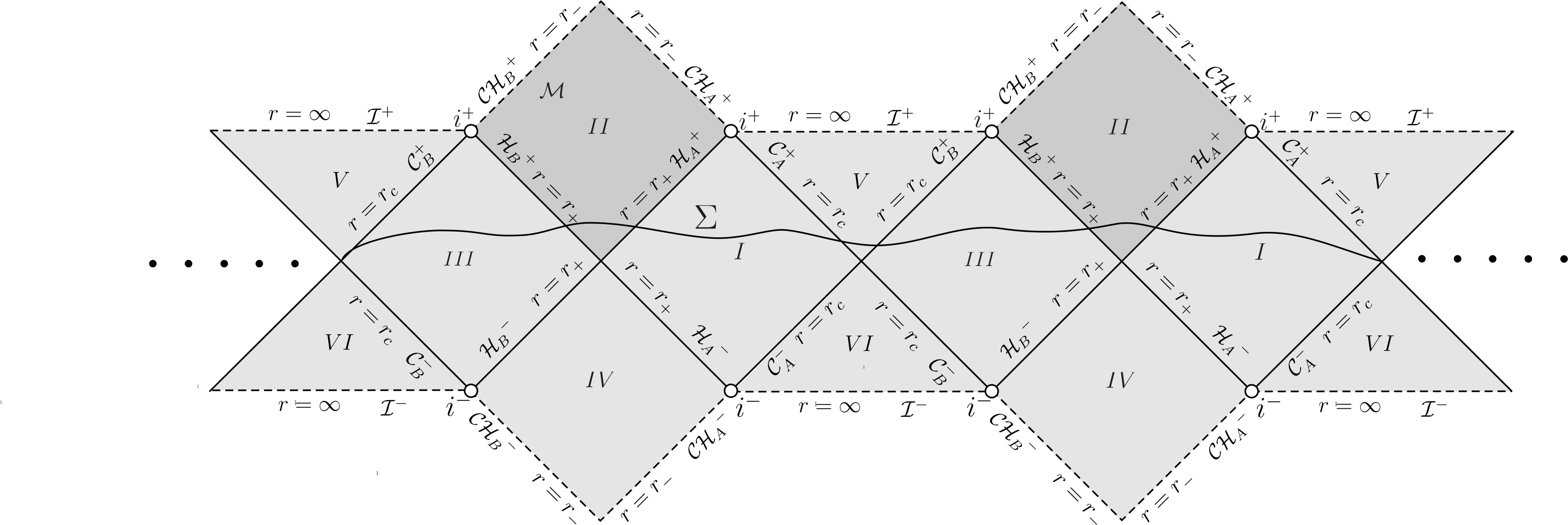}
\caption[Text der im Bilderverzeichnis auftaucht]{Conformal diagram of Reissner-Nordstr\"om-de Sitter spacetime.}
\label{RN_ganz}\end{figure}}
The region of interest for this paper is the darker shaded region $II$ (of Figure \ref{RN_ganz}); since all of its connected components are isometric
we choose one and denote it by \mbox{$\cQ_{int}$}.
We can then recover (the corresponding connected component of) the black hole interior region by \mbox{$\cM_{int}=\pi^{-1}(\cQ_{int})$},
where $\pi$ is the projection \mbox{$\pi: \cM\rightarrow \cQ$}.

 In the usual Schwarzschild coordinates $(t,r, \omega)\in\mathbb{R}\times(r_-,r_+)\times \mathbb{S}^2=\cM_{int}$ the metric takes the form
\begin{equation}
 \lb{metricTR}
g=-D(r) d t^2 +\frac{1}{D(r)} d r^2 +r^2 d\sigma^2_{\mathbb{S}^2}\;,
\end{equation}
where $d\sigma^2_{\mathbb{S}^2}$ is the round metric of the 2-sphere and the static potential is defined by
\bea
\lb{static potential}
D(r)=1-\frac{2M}{r}+\frac{e^2}{r^2}-\frac{\Lambda }{3}r^2 \;,
\eea
where the (cosmological) black hole parameters are: the mass $M>0$, the charge $e\neq 0$  and the cosmological constant $\Lambda>0$.
Note that $T=\partial_t$ is Killing and is known as the stationary vector field.

The subextremality assumption follows by simply requiring the existence of 3 positive and distinct roots
$$0<r_-<r_+<r_c$$
of the potential $D$; to see how this can be encoded in terms of the black hole parameters see~\cite{joao3}[Appendix A]. The locus $r=r_+$ is known as the
event horizon $\cH^{\pm}$
and $r=r_-$, which in the projective picture corresponds to the boundary of the Penrose diagram $\cal Q$ in $\mathbb{R}^{1+1}$,
will be referred to as the Cauchy horizon ${\cal CH}^{\pm}$, even without any mention to an extension of $(\cM,g)$.

 The following quantities, known as surface gravities~\footnote{Observe that, by definition, our surface gravities are non-negative.},
 will play a fundamental role in our work:
\begin{equation}
\label{surfGrav}
\kappa_*=\frac{1}{2}\left|D'(r_*)\right|\quad,\quad *\in\{-,+,c\}\;.
\end{equation}
Note that the subextremality condition is equivalent to having non-vanishing surface gravities.

Consider the tortoise coordinate $r^*$ determined by
\begin{equation}
 \label{r*}
\left\{
\begin{array}{l}
  \frac{dr^*}{dr}=\frac{1}{D} \\
 r^*(r_0)=0\;,
\end{array}
\right.
\end{equation}
for an $r_0\in(r_-,r_+)$ to be chosen at our later convenience. This allows us to construct
the coordinates
\begin{equation}
 \label{EFv}
\left\{
\begin{array}{l}
  v=t+r^*(r) \\
  r=r \\
  \omega=\omega\;,
\end{array}
\right.
\end{equation}
which will be specially helpful near the event horizon. In fact, in these coordinates, the metric takes the form
\bea
\lb{v_metric}
g=-D\md v^2+2\md v \md r+ r^2d\sigma^2_{\mathbb{S}^2}\;,
\eea
and extends regularly to
the manifold with boundary
$$\cM_{int}\cup{\cal H}_A^+\cup{\cal CH}^+_B,$$
see Figures~\ref{RN_ganz} and~\ref{range}, with $-\partial_r$ future directed and transverse  to  $\cH_A^+$ and $T=\partial_v$.

Near the Cauchy horizon it will be convenient to  consider instead the coordinates
\begin{equation}
 \label{EFu}
\left\{
\begin{array}{l}
  u=t-r^*(r) \\
  \hat r=r \\
  \omega=\omega\;,
\end{array}
\right.
\end{equation}
under which the metric takes the form
\bea
\lb{u_metric}
g=-D\md u^2-2\md u \md \hat r+ \hat r^2d\sigma^2_{\mathbb{S}^2}\;.
\eea
It is essential to note that in these coordinates the metric extends regularly to the manifold with boundary
$$\cM_{int}\cup{\cal CH}_A^+\cup{\cal H}_B^+,$$
with $-\partial_{\hat r}$ future directed and transverse to  ${\cal CH}_A^+$. Observe also that the orientation of $u$ is contrary to the
spacetime orientation and as a consequence $\cH_A^+\supset\{u=+\infty\}$.

The extra care in giving different notations for the radial coordinate in each set of coordinates is to protect us from the
{\em fundamental confusion of calculus} which can be specially pernicious in our setting where it is essential to distinguish between
transverse and tangential directions at the Cauchy horizon.

Integrating~\eqref{r*} gives
$$r^*=\frac{1}{2\kappa_+}\ln{|{r-r_+}|}-\frac{1}{2\kappa_-}\ln{|{r-r_-}|}-\frac{1}{2\kappa_c}\ln{|{r-r_c}|}+\frac{1}{2\kappa_n}\ln{|{r-r_n}|}+C\;,$$
where $r_n$ is the negative root of $D$. So, since
\begin{equation}
\label{vur}
 v-u=2r^*\;,
\end{equation}
we see that near $r_+$ we have the estimate
\begin{equation}
\label{r-r}
r_+-r(u,v) = e^{O(1)} e^{\kappa_+(v-u)}\,.
\end{equation}
We also take the chance to use~\eqref{vur} and define
\begin{equation}
 v_r(u):=u+2r^*(r)\;,
\end{equation}
and
\begin{equation}
u_r(v):=v-2r^*(r)\;.
\end{equation}
Following~\cite{jan1}, we will use the following notations for the level sets of the coordinate functions $u,v$ and $r$:
\begin{eqnarray*}
\label{levelSets}
\cC_{v_1}=\{v=v_1\}\;, \\
{\cal C}_{u_1}=\{u=u_1\}\;, \\
\Sigma_{r_1}=\{r=r_1\} \;.
\end{eqnarray*}
Given $S\subset \cM_{int}$, a function $x:\cM_{int}\rightarrow \mathbb{R}$, we will also write
$$S(x_1,x_2):=\{p\in\cM_{int} \;|\; p\in S\text{ and } x_1\leq x(p)\leq x_2\}\;;$$
for instance with the previous notations
$$\cC_{v}(u_1,u_2)=\{p\in\cM_{int} \;|\; v(p)=v_1\text{ and } u_1\leq u(p)\leq u_2\}\;.$$

 {\begin{figure}[ht]
\centering
\includegraphics[width=0.3\textwidth]{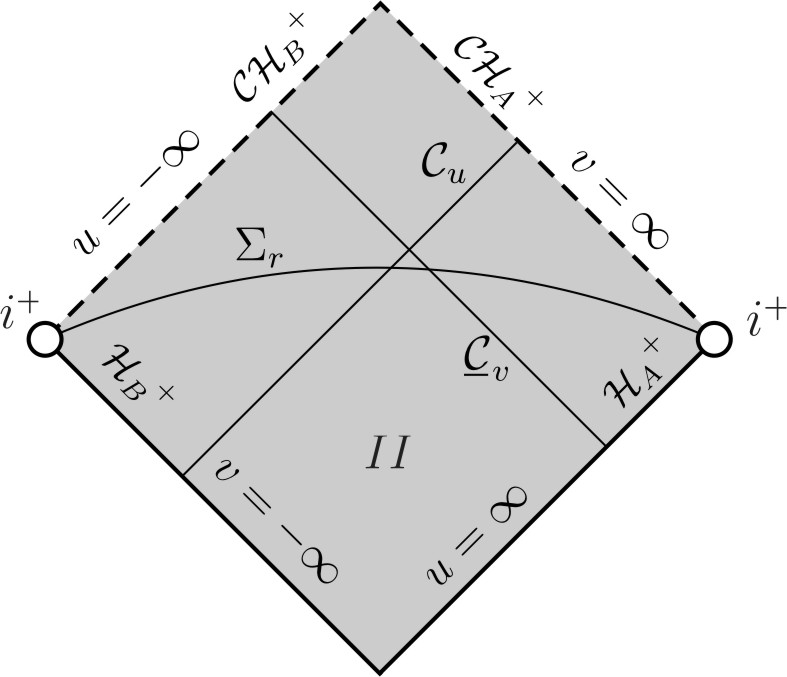}
\caption[Text der im Bilderverzeichnis auftaucht]{Conformal diagram of $\cQ_{int}=\pi(\cM_{int})$
with the ranges of $u$ and $v$, and level sets of $u,v$ and $r$ depicted.}
\label{range}
\end{figure}}

\subsection{Energy currents and vector fields}
\lb{en_cur}

To set notation and to try to make this paper as self--contained as possible, we will briefly review some of the basics of the vector field method;
for a more detailed discussion see~\cite{sergiu,m_lec} and  references therein.

Define the stress-energy tensor of a massless scalar field by
\ben
\lb{energymomentum}
T_{\mu\nu}[\phi]=\partial_\mu\phi\partial_\nu\phi-\frac12g_{\mu\nu}g^{\alpha\beta}\partial_{\alpha}\phi\;\partial_\beta\phi\;.
\een
If $\phi$ is a solution to the wave equation~\eqref{waveEq} then
we obtain the energy-momentum conservation law
\begin{equation}
\label{divfree}
\nabla^\mu T_{\mu\nu}=0\;.
\end{equation}
By contracting the energy-momentum tensor with a vector field $V$, we define the current
\be
\lb{J}
J_\mu^V[\phi]:= T_{\mu\nu}[\phi] V^\nu.
\ee
In this context we call $V$ a multiplier vector field.

If $V$ and $n$ are both causal and future directed we have
\begin{equation}
 \label{signCurrent}
 J^V_\mu [\phi]n^\mu\geq 0\;.
\end{equation}

In the vector field method the basic energy estimates follow by applying Stokes' Theorem to the divergence of currents.
By~\eqref{energymomentum} it follows that
\bea
\nabla^{\mu}J^V_{\mu}[\phi]=K^V[\phi]\;,
\eea
for
\be
\lb{K}
K^V[\phi]:= (\pi^V)^{\mu\nu}T_{\mu\nu}[\phi]\;,
\ee
where $(\pi^V)^{\mu\nu} :=\frac{1}{2} (\cL_V g)^{\mu\nu}$ is the deformation tensor of $V$.
With these definitions it is immediately clear that $\nabla^\mu J^V_\mu[\phi] =0$, if $V$ is Killing.

In this language, including the notation~\eqref{levelSets}, Stokes' Theorem applied to the divergence of the current
$J^V_{\mu}[\phi]$, in the region depicted in Figure~\ref{start}
{\begin{figure}[ht]
\centering
\includegraphics[width=0.4\textwidth]{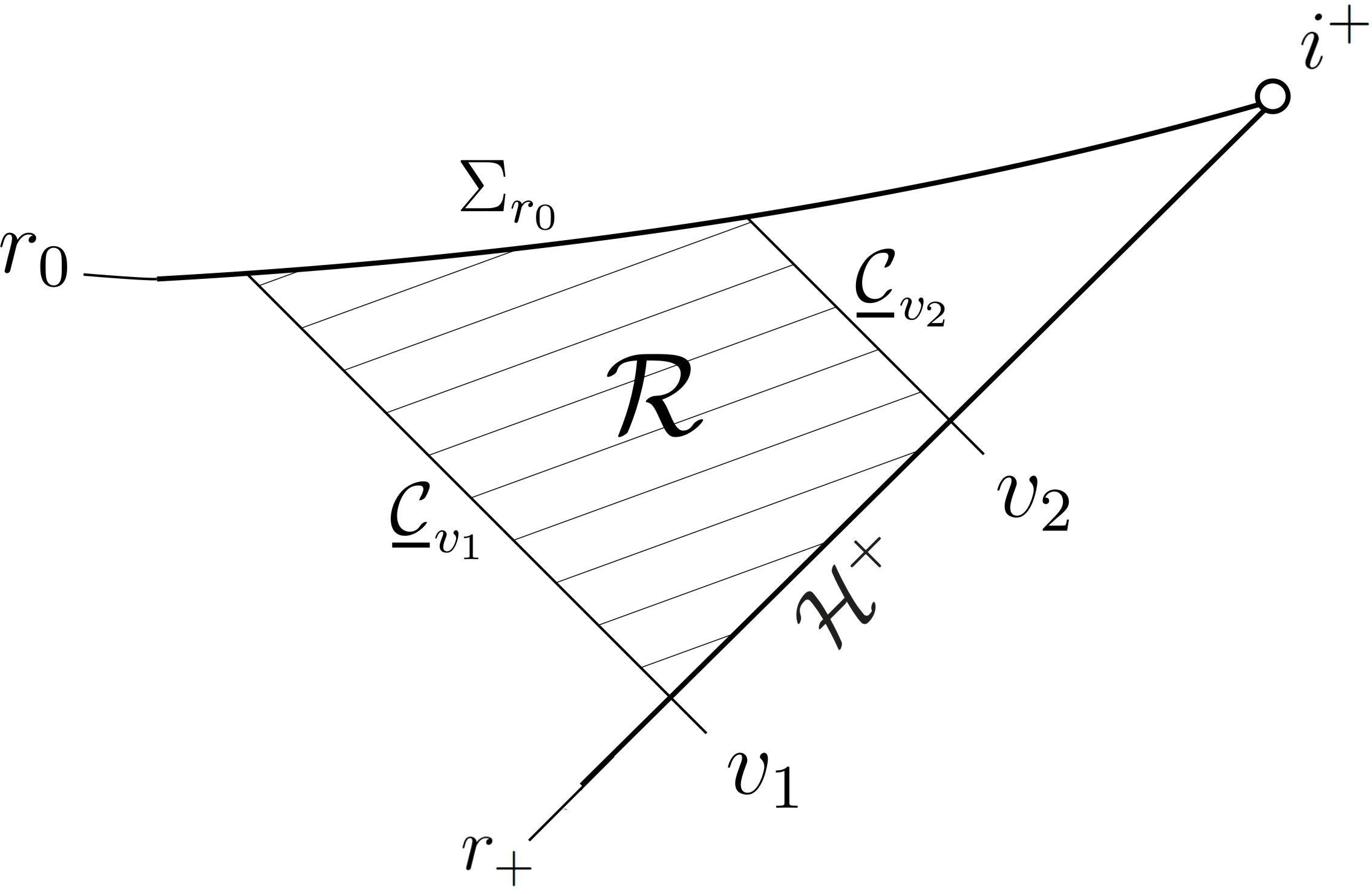}
\caption[]{Region to apply divergence theorem represented as the hatched area.}
\label{start}\end{figure}}
provides the energy identity
\begin{eqnarray}
\label{enEst1}
\nonumber
\int\limits_ {{\cC}_{v_2}(r_0,r_+)} J_{\mu}^V[\phi] n^{\mu}_{\cC_{v}}\dVt_{\cC_{v}}
&+&
\int\limits_ {\Sigma_{r_0}(v_1,v_2)} J_{\mu}^V[\phi] n^{\mu}_{\Sigma_{r_0}} \dVt_{\Sigma_{r_0}}+\int\limits_ {\cR} K^V[\phi] \dV
\\
&=&
\int\limits_ {{\cC}_{v_1}(r_0,r_+)} J_{\mu}^V[\phi] n^{\mu}_{\cC_{v}} \dVt_{\cC_{v}} +
\int\limits_{\cH^+(v_1,v_2)}J_{\mu}^V[\phi]n^{\mu}_{\cH^+}\dVt_{\cH^+}\;.
\end{eqnarray}

To use the previous, some clarifications are still required. In the coordinates~\eqref{EFv}:
\begin{itemize}
 \item The volume element associated to $g$ is
 $$\dV = r^2dv dr d\omega\;,$$
 where, from now on $d\omega$ is the volume form of the round 2-sphere.
 \item On each $\Sr$ we have the future directed unit normal
$$n_{\Sr}=\frac{1}{\sqrt{-D}}\left( \partial_v+D\partial_r\right)\;,$$
and the induced volume form
$$\dVt_{\Sr}=n^{\mu}_{\Sr}\neg\dV= r^2\sqrt{-D} dv  d\omega\;.$$
\item On the null components
of $\partial\cR$, say $\cC_v$, there is no natural choice of normal or volume form, so one can just choose $n_{\cC_{v}}$ to be any
future directed vector orthogonal to $\cC_v$  and then $\dVt_{\cC_{v}}$ is completely determined by Stokes' Theorem; for instance for the choice
$$n^{\mu}_{\cC_{v}}=-\partial_r\;,$$
corresponds
$$\dVt_{\cC_{v}}=r^2dr d\omega\;.$$
\end{itemize}

\bigskip

A specially important vector field for the analysis of waves in black hole spacetimes is the red-shift vector field $N$ discovered
by Dafermos and Rodnianski~\cite{m_red}, which captures the decay generated by the red-shift effect near the event horizon.
In the present paper we will need the following particular case of such construction, whose proof can be found in~\cite{stef_lec}:
\begin{thm}
\label{redshiftThm}
Let $\phi$ be a solution of the wave equation. Given $\epsilon>0$, there exists $r_0\in(r_-,r_+)$ and a future directed vector field
$N=N^r(r)\partial_r+N^v(r)\partial_v$  such that, in $r_0\leq r\leq r_+$, the following estimate holds
\begin{equation}
\label{redEst}
 K^N[\phi]\geq (\kappa_+-\epsilon) J_{\mu}^N[\phi] n^{\mu}_{\cC_{v}}\;.
\end{equation}
\end{thm}


\subsection{Killing vector fields}
\lb{angular}

In the context of the vector field method, Killing vector fields are specially useful not just as multipliers,
as we already saw, but also as commutators: for a Killing vector field $W$ we have the commutation relation
$[\Box_g, W]=0$, so if $\phi$ is a solution of the wave equation, then $W\phi$ is also a solution.

On Reissner-Nordstr\"om besides the stationary vector field $T$ we also have the Killing vectors provided by
the generators of spherical symmetry $\leo_{i}$, $i=1,2,3$. These satisfy the important relations~\cite{anne}
\begin{equation}
\label{roundGrad}
 |\nabb \phi|^2=\frac{1}{r^2}\sum_{i=1}^{3} \left(\leo_i \phi\right)^2\;,
\end{equation}
and
\begin{equation}
\label{roundLap}
 \left(\Dell \phi \right)^2=\frac{1}{r^4}\sum_{i=1}^{3} \sum_{j=1}^{3}\left(\leo_i \leo_j\phi\right)^2\;,
\end{equation}
where $\nabb$ and $\Dell$ are the gradient and Laplacian of the metric $\gSlash=r^2d\sigma^2_{\mathbb{S}^2}$\;.

\subsection{Price's Law}
\lb{secPrice}

We are now able to formalize our main assumption:

\begin{ass}[Price's Law]
\lb{pricesAss}
We will say that a function $\phi\in C^{\infty}(\Mint\cup \cH^+)$ satisfies Price's Law,
on Reissner--Nordstr\"om--de Sitter, provided there exists $p>0$ such that for any function of the form
\bea
\label{psi}
\psi=T^l\leo^I\phi\;,
\eea
constructed out of any $l\in\mathbb{N}_0$ and $I=(i_1,...,i_m)$, with $m\in\mathbb{N}_0$ and $i_s\in\{1,2,3\}$,
we have, in $\cH^+=\{r=r_+\}$,
covered by the Eddington-Finkelstein coordinates~\eqref{EFv},
\bea
\lb{Price}
\int_{v_1}^{v_2}\int_{\bbS^2}\left[ \left(\partial_v\psi\right)^2
+\left|{\nabb}\psi\right|^2\right](v,r_+,\omega)d\omega dv
&\lesssim_{l,I}& e^{-2p v_1}\;,
\eea
for all $0\leq v_1\leq v_2$.
\end{ass}



\section{The main result}
\lb{mainresult}

The main result of this paper, the proof of which will be given in Section~\ref{secProof}, can be stated as follows:
\begin{thm}
\lb{mainThm}
Let $(\Mint, g)$ be a black hole interior region of a subextremal Reissner--Nordstr\"om--de Sitter spacetime.
Let $\phi\in C^{\infty}(\Mint\cup \cH^+)$ be a solution of the wave equation~\eqref{waveEq} satisfying
Price's Law (Assumption~\ref{pricesAss}). If (recall~\eqref{surfGrav})
\begin{equation}
\label{mainCond}
2\min\{p,\kappa_+\}>\kappa_-\;,
\end{equation}
then for any timelike and future directed vector field $N\in {\cal X}^{\infty}(\Mint\cup {\cal CH}^+)$ that commutes
with the stationary Killing vector field $T$,  there exists a constant $C>0$ such that, for any
$r_-\leq r_2\leq r_1\leq r_+$ and any $0\leq u_1\leq u_2$, we have (compare with~\eqref{mainEstCoords})
\begin{equation}
\label{mainEst}
\int\limits_ {{\cal C}_{u_1}(r_2,r_1)}  J_{\mu}^{ N}[\phi] n^{\mu}_{{\cal C}_u}\dVt_{{\cal C}_u}
+ \int\limits_{\Sigma_{r_2}(u_1,u_2)}  J_{\mu}^{ N}[\phi] n^{\mu}_{\Sigma_r}\dVt_{\Sr}
\leq C\;,
\end{equation}
and, in $\Mint$,
\begin{equation}
\label{mainBound}
\left|\phi\right|
\leq C\;.
\end{equation}
Moreover, $\phi$ extends as a function to $C^{0}(\Mint\cup\cH^+ \cup{\cal CH}^+)\cap H_{loc}^1(\Mint\cup\cH^+\cup{\cal CH}^+)$.
\end{thm}
\begin{rmk}
It is instructive to notice that in the Eddington-Finkelstein coordinates~\eqref{EFu},
which recall are regular up to and including ${\cal CH}^+$ with $\partial_{\hat r}$ transverse to the Cauchy horizon, the estimate~\eqref{mainEst} becomes
\begin{equation}
\label{mainEstCoords}
\int_{r_2}^{r_1}\int_{\mathbb{S}^2}
\left[\left(\partial_{\hat r}\phi\right)^2+\left|{\nabb}\phi\right|^2\right](u_1,r,\omega)d\omega dr
+
\int_{u_1}^{u_2}\int_{\mathbb{S}^2}
\left[\left(\partial_u\phi\right)^2+(-D)\left(\partial_{\hat r}\phi\right)^2+\left|{\nabb}\phi\right|^2\right](u,r_2,\omega)d\omega du
\leq C\;.
\end{equation}
\end{rmk}

\begin{rmk}
In fact, if we restrict to the region $r_-\leq r_2\leq r_1\leq \check r_1$, for a $\check r_1$ sufficiently close to $r_-$,
we get some decaying (in $u$) statements: see~\eqref{lastEnergyEst},~\eqref{finalDecay1} and~\eqref{finalDecay2}.
\end{rmk}

\begin{rmk}
The statement for $r_2=r_-$ should be interpreted as referring to
$$\int\limits_{\Sigma_{r_-}(u_1,u_2)}  J_{\mu}^{ N}[\phi] n^{\mu}_{\Sigma_r}\dVt_{\Sr}:=
\limsup_{r\rightarrow r_- } \int\limits_{\Sigma_{r}(u_1,u_2)}  J_{\mu}^{ N}[\phi] n^{\mu}_{\Sigma_r}\dVt_{\Sr}\;.$$
\end{rmk}

\begin{rmk}
\label{rmkRN}
If one keeps Assumption~\ref{pricesAss} as it is stated here (with the same exponential decay in those coordinates), then Theorem~\ref{mainThm}
will also hold on the black hole interior of subextremal (asymptotically flat) Reissner--Nordstr\"om or Reissner--Nordstr\"om--adS spacetimes.
In fact the proof presented here will go through with minimal changes. Note also that the reason to present the result in the de Sitter context is that,
for generic Cauchy data (plus appropriate boundary data if $\Lambda<0$),
Assumption~\ref{pricesAss}  is not expected to hold if $\Lambda\leq0$.
\end{rmk}

\begin{rmk}
We expect our main condition~\eqref{mainCond} to be sharp. More precisely, it is expected, for instance in view of the numerology
of mass inflation~\cite{brady2,joao3}, that all solutions of the wave equation that, instead of Assumption~\ref{pricesAss}, satisfy the stronger condition
\bea
\lb{Price}
\int_{v_1}^{v_2}\int_{\bbS^2}\left[ \left(\partial_v\psi\right)^2
+\left|{\nabb}\psi\right|^2\right](v,r_+,\omega)d\omega dv
&\sim_{l,I}& e^{-2p v_1}\;,
\eea
will have unbounded $H^1$ norm near the Cauchy horizon, provided~\eqref{mainCond} does not hold.
\end{rmk}

%



\section{Proof of the Theorem~\ref{mainThm}}
\label{secProof}

\subsection{Red-shift estimates}
\label{secRed}

We start by exploring the red-shift vector field:
\begin{prop}
\label{propDecay0}
Let $\phi$ be a solution of the wave equation~\eqref{waveEq} that satisfies Price's Law (Assumption~\ref{pricesAss}) and let $N$ be the red-shift vector
field constructed in Theorem~\ref{redshiftThm}.
Then, for $\kappa<\min\{\kappa_+,2p\}$ and any function of the form~\eqref{psi} there exists $r_0 \in (r_{-}, r_+)$, such that, for all $r\in[r_0,r_+]$,
$\omega\in\bbS^2$ and $0\leq v_1 \leq v_2$ the following estimates hold
\bea
\lb{decayC}
\int\limits_ {{\cC}_{v_1}(r_0,r_+)} J_{\mu}^N[\psi] n^{\mu}_{\cC_{v}} \dVt_{\cC_{v}}
\sim
\int_{r_0}^{r_+}\int_{\bbS^2}
\left[\left(\partial_r\psi\right)^2+\left|{\nabb}\psi\right|^2\right](v_1,r,\omega) d\omega dr
\lesssim_{l,I} e^{-\kappa {v_1}}\; ,
\eea
\bea
\lb{decayS}
\int\limits_ {\Sigma_r(v_1,v_2)} J_{\mu}^N[\psi] n^{\mu}_{\Sigma_r} \dVt_{\Sigma_r}
\sim
\int_{v_1}^{v_2}\int_{\bbS^2}
\left[\left(\partial_v\psi\right)^2+(-D)\left(\partial_r\psi\right)^2+\left|{\nabb}\psi\right|^2\right](v,r,\omega)d\omega dv
\lesssim_{l,I} e^{-\kappa {v_1}}\; ,
\eea
\bea
\lb{decayR}
\int\limits_ {\cR(v_1,v_2)} K^N[\psi] \dV
\sim
\int_{v_1}^{v_2}\int_{r_0}^{r_+}\int_{\bbS^2}
\left[\left(\partial_v\psi\right)^2+\left(\partial_r\psi\right)^2+\left|{\nabb}\psi\right|^2\right](v,r,\omega)d\omega dr dv
\lesssim_{l,I} e^{-\kappa {v_1}}\;.
\eea
\end{prop}

\begin{proof}
 The divergence Theorem  applied to the current $J_{\mu}^N[\psi]$ in the region $\cR(v_1, v_2)$ defined by Figure~\ref{start}, in Section~\ref{en_cur}, provides the energy identity already stated in \eqref{enEst1}.
%
Since, by~\eqref{signCurrent} and~\eqref{redEst},  all the terms in the previous identity are non negative, we see that claims~\eqref{decayS}
and~\eqref{decayR}
follow immediately from~\eqref{decayC} and the assumption~\eqref{Price}. So, we only need to establish~\eqref{decayC}.
Notice also that since ${\cC}_{v_2}(r,r_+)\subset{\cC}_{v_2}(r_0,r_+)$, for all $r\geq r_0$, we just need to consider the case $r=r_0$.

Following~\cite{anne} we define
$$f(v):=\int\limits_ {{\cC}_{v}(r_0,r_+)} J_{\mu}^N[\psi] n^{\mu}_{\cC_{v}}\dVt_{\cC_{v}}\;.$$
Relying on Theorem~\ref{redshiftThm}, for any $\epsilon>0$, we can choose $r_0$ sufficiently close to $r_+$ such that
\begin{eqnarray*}
\int\limits_ {\cR(v_1,v_2)} K^N[\psi] \dV
&=&
\int_{v_1}^{v_2}\int_{r_0}^{r_+}\int_{\bbS^2} K^N[\psi] r^2 d\omega\, dr\,dv
\\
&\geq&
(\kappa_+-\epsilon)\int_{v_1}^{v_2}\left(\int_{r_0}^{r_+}\int_{\bbS^2}  J_{\mu}^N[\psi] n^{\mu}_{\cC_{v}} r^2 d\omega\, dr\right)\,dv
\\
&=&
(\kappa_+-\epsilon)\int_{v_1}^{v_2} f(v)dv\;.
\end{eqnarray*}
Then, dropping the second term from the left-hand side of~\eqref{enEst1} while using~\eqref{Price} one obtains
$$f(v_2)+ (\kappa_+-\epsilon)\int_{v_1}^{v_2} f(v)dv \leq f(v_1)+ Ce^{-2p\,v_1}\;,$$
for all $0 \leq v_1 \leq v_2$. Claim~\eqref{decayC} is then a consequence of the following general result:

\begin{lem}
\label{decayLemma}
For $t_0\geq 0$, let $f$ be a continuous function, $f:[t_0, \infty)\rightarrow \bbR^+$, satisfying
\bea
\lb{starter3}
f(t_2) + b\int_{t_1}^{t_2} f(t)\md {t} \leq N\,f(t_1) +B\,e^{-\Delta {t_1}},
\eea
for given constants $b,N,\Delta,B>0$ and all $0\leq {t_1}\leq t_2$.

Then, for any $\kappa<\min\{b/N, \Delta\}$ we have
\bea
\lb{ftwant}
f(t) \leq C_{\kappa} e^{-\kappa\,{t}},
\eea
for all $t\geq 0$, where $C_{\kappa}>0$ is a constant that might blow up as $\kappa$ approaches $\min\{b/N, \Delta\}$.
\end{lem}
\begin{proof}
In particular we have, for all $t\leq t_2$,
$$f(t_2)\leq Nf(t)+Be^{-\Delta {t}},$$
which integrated on any interval of the form $[t_1,t_2]$ implies
$$(t_2-t_1)f(t_2)\leq N\left[\frac{N}{b}f(t_1)+\frac{B}{b}e^{-\Delta {t_1}}\right]+\frac{B}{\Delta}e^{-\Delta {t_1}}\;.$$
Setting $k=\min\{b/N, \Delta\}$ we then see that, for all $t_1\leq t_2$,
$$(t_2-t_1)f(t_2)\leq \frac{N}{k}f(t_1)+\frac{2B}{k}e^{-\Delta {t_1}}\;.$$
If we keep on integrating, it follows by induction that, for all $n\in\bbN_0$ and all $t_1\leq t_2$,
$$\frac{(t_2-t_1)^n}{n!}f(t_2)\leq \frac{N}{k^n}f(t_1)+\frac{(n+1)B}{k^n}e^{-\Delta {t_1}}\;.$$
Setting $t=t_2$ and $t_1=t_0$ in the previous inequality we see that
$$\frac{k^n(t-t_0)^n}{n!}f(t)\leq Nf(t_0)+B(n+1)e^{-\Delta {t_0}}\;.$$
Since the right-hand side is not summable, we choose $0<\epsilon<1$ and multiply both sides by $(1-\epsilon)^{2n}$. We then use the fact that
$(1-\epsilon)^n(n+1)\leq C(\epsilon)$, $\forall n$, to conclude that
$$\frac{\left((1-\epsilon)^2k(t-t_0)\right)^n}{n!}f(t)\leq C(\epsilon)(1-\epsilon)^n\;.$$
Summing both sides leads to the desired result presented in the form:
$$e^{(1-\epsilon)^2k(t-t_0)}f(t)\leq \frac{C(\epsilon)}{\epsilon}\;.$$
\end{proof}
\end{proof}


\subsection{Conditional pointwise estimates}
\label{sectionCond}

We will now see how an energy estimate of the form~\eqref{decayS} implies pointwise estimates for the current
$$J_{\mu}^N[\psi] n^{\mu}_{\cC_{v}}\sim \left(\partial_r\psi\right)^2+\left|{\nabb}\psi\right|^2\;.$$

At the moment, we only have~\eqref{decayS} for $\kappa<\min\{\kappa_+,2p\}$ but,
later on, we will be able to improve this in
a iterative process where the estimates proven in this section will play a crucial role;
therefore, in anticipation of such fact,
we will present our estimates under the following assumption:

\begin{ass}
\lb{energyAss}
For a given $\bar{\kappa}>0$, any $\kappa<\bar{\kappa}$ and any function of the form~\eqref{psi},
there exists $r_0 \in (r_{-}, r_+)$ such that, for all $r\in[r_0,r_+]$ and all $0\leq v_1\leq v_2$, the following holds
\bea
\lb{decaySCond}
\int\limits_ {\Sigma_r(v_1,v_2)} J_{\mu}^N[\psi] n^{\mu}_{\Sr} \dVt_{\Sigma_r}
\sim
\int_{v_1}^{v_2}\int_{\bbS^2}
\left[\left(\partial_v\psi\right)^2+(-D)\left(\partial_r\psi\right)^2+\left|{\nabb}\psi\right|^2\right](v,r,\omega)d\omega dv
\lesssim_{l,I}  e^{-\kappa\, {v_1}}\;.
\eea
\end{ass}
Our main goal in this section is to prove the following:

\begin{prop}
\label{propDecayJ}
Let $\phi$ be a solution of~\eqref{waveEq} that satisfies Assumption~\ref{energyAss}. Then, for any $\kappa<\min\{\bar\kappa, 2\kappa_+\}$,
there exists $r_0 \in (r_{-}, r_+)$, such that,
for all $r\in[r_0,r_+]$, $\omega\in\bbS^2$ and $v\geq 0$,
\bea
\label{decayJ}
J_{\mu}^N[\psi] n^{\mu}_{\cC_{v}}(v,r,\omega)\lesssim e^{-\kappa v}\;.
\eea
\end{prop}

\begin{proof}
The desired result follows from Proposition~\ref{propPointDecayVO} and
Proposition~\ref{propPointDecayR}.\\
\end{proof}


We will need the following basic result:

\begin{lem}
\label{lemExpDecay}
 Let $f:[t_0,+\infty[\rightarrow \bbR$ be a function for which there exists $\kappa_1,C>0$ such that, for all $t\geq t_0$,
$$\int_t^{\infty}f(\bar t) d\bar t\leq C e^{-\kappa_1t} \;.$$
Then, for all $0<\kappa_2<\kappa_1$,
$$\int_t^{\infty}e^{\kappa_2\bar t}f(\bar t) d\bar t\leq C e^{-(\kappa_1-\kappa_2) t} \;.$$
\end{lem}
\begin{proof}
Recall that the {\em floor} function is defined by $\lfloor t \rfloor:=\sup_{z\in\mathbb{Z}}\{z\leq t\}$ and trivially satisfies
$t-1\leq \lfloor t \rfloor\leq t$. Then

\begin{eqnarray*}
\int_t^{\infty}e^{\kappa_2\bar t}f(\bar t) d\bar t
&\leq&
\sum_{n= \lfloor t \rfloor}^{\infty} \int_n^{n+1} e^{\kappa_2\bar t}f(\bar t)d\bar t
\\
&\leq&
\sum_{n= \lfloor t \rfloor}^{\infty} e^{\kappa_2(n+1)} \int_n^{n+1} f(\bar t)d\bar t
\\
&\leq&
C\sum_{n= \lfloor t \rfloor}^{\infty} e^{-(\kappa_1-\kappa_2) n}
\\
&\leq&
C\sum_{n= \lfloor t \rfloor}^{\infty} \left(e^{-(\kappa_1-\kappa_2)} \right)^n
\\
&\leq&
C\frac{e^{-(\kappa_1-\kappa_2) \lfloor t \rfloor}}{1-e^{-(\kappa_1-\kappa_2)}}
\\
&\leq&
Ce^{-(\kappa_1-\kappa_2)t}\;.
\end{eqnarray*}
\end{proof}


Anticipating the use of Sobolev embedding over the symmetry spheres of the background manifold we note that:

\begin{lem}
\label{h2s2}
 Under Assumption~\ref{energyAss} we have, for all $i\in\{1,2,3\}$,
\bea
\int_{v_1}^{v_2} \|\partial_v\psi(v,r,\,\cdot\,)\|^2_{H^2(\bbS^2)}dv\lesssim e^{-\kappa v_1}\;,
\eea
\bea
\int_{v_1}^{v_2} \|\leo_i\psi(v,r,\,\cdot\,)\|^2_{H^2(\bbS^2)}dv\lesssim e^{-\kappa v_1}\;,
\eea
and, consequently,
\bea
\int_{v_1}^{v_2} e^{\kappa v}\|\partial_v\psi(v,r,\,\cdot\,)\|^2_{H^2(\bbS^2)}dv\lesssim 1\;,
\eea
and
\bea
\int_{v_1}^{v_2} e^{\kappa v} \|\leo_i\psi(v,r,\,\cdot\,)\|^2_{H^2(\bbS^2)}dv\lesssim 1\;.
\eea
\end{lem}
\begin{proof}
The first estimate follows by an application of the Assumption~\ref{energyAss} to
\begin{eqnarray*}
 \int_{v_1}^{v_2}\|\partial_v\psi(v,r,\,\cdot\,)\|^2_{H^2(\bbS^2)} dv
&\lesssim&
\int\limits_ {\Sigma_r(v_1,v_2)} J_{\mu}^N[\psi] n^{\mu}_{\Sigma_r} \dVt_{\Sigma_r}
+
\sum_{i=0}^3\int\limits_ {\Sigma_r(v_1,v_2)} J_{\mu}^N[\leo_i\psi] n^{\mu}_{\Sigma_r} \dVt_{\Sigma_r}
\\
&+&
\sum_{i=0}^3\sum_{j=0}^3\int\limits_ {\Sigma_r(v_1,v_2)} J_{\mu}^N[\leo_i\leo_j\psi] n^{\mu}_{\Sigma_r} \dVt_{\Sigma_r} \lesssim e^{-\kappa v}\;.
\end{eqnarray*}
The second follows in a similarly fashion by using~\eqref{roundGrad} and boundedness of $r$, in our region of interest.
The last two are then an immediate consequence of Lemma~\ref{lemExpDecay}.
\end{proof}

We are now able to obtain our first pointwise bounds:

\begin{prop}
\label{propPointDecayVO}
Under Assumption~\ref{energyAss}
\bea
\lb{pointDecayV}
\left(\partial_v\psi(v,r,\omega)\right)^2\lesssim e^{-\kappa\, v}\;,
\eea
and, for all $i\in\{1,2,3\}$,
\bea
\lb{pointDecayO}
\left(\leo_i\psi(v,r,\omega)\right)^2\lesssim e^{-\kappa\, v}\;.
\eea
\end{prop}

\begin{proof}
Squaring
\begin{eqnarray*}
\zeta(v_2,r,\omega)=\int_{v_1}^{v_2}\partial_v\zeta(v,r,\omega)dv + \zeta(v_1,r,\omega)\;,
\end{eqnarray*}
where $\zeta$ is an arbitrary function to be specialized later,
leads to
\begin{eqnarray*}
\zeta^2(v_2,r,\omega)
&\leq&
2 \left(\int_{v_1}^{v_2}\partial_v\zeta(v,r,\omega)dv\right)^2 + 2\zeta^2(v_1,r,\omega)
\\
&=&
2 \left(\int_{v_1}^{v_2}e^{-\frac{1}{2}\kappa v}e^{\frac{1}{2}\kappa v}\partial_v\zeta(v,r,\omega)dv\right)^2
+ 2\zeta^2(v_1,r,\omega)
\\
&\leq&
2 \int_{v_1}^{v_2}e^{-\kappa v}dv \int_{v_1}^{v_2}e^{\kappa v}\left(\partial_v\zeta(v,r,\omega)\right)^2dv
+ 2\zeta^2(v_1,r,\omega)
\\
&\leq&
C e^{-\kappa v_1} \int_{v_1}^{v_2}e^{\kappa v}\left(\partial_v\zeta(v,r,\omega)\right)^2dv + 2\zeta^2(v_1,r,\omega)\;.
\end{eqnarray*}
Setting
$$Z(v,r):=\sup_{\omega\in\bbS^2} \zeta^2(v,r,\omega)\;,$$
the previous inequality gives, by Sobolev embedding in $\mathbb{S}^2$,
\begin{eqnarray*}
Z(v_2,r)
&\leq&
Ce^{-\kappa v_1}\int_{v_1}^{v_2}e^{\kappa v}\sup_{\omega\in\bbS^2}\left(\partial_v\zeta(v,r,\omega)\right)^2dv+2Z(v_1,r)
\\
&\leq&
Ce^{-\kappa v_1}\int_{v_1}^{v_2}e^{\kappa v}\|\partial_v\zeta(v,r,\,\cdot\,)\|_{H^2(\bbS^2)}^2dv+2Z(v_1,r)\;.
\end{eqnarray*}

Now, if $\zeta$ is of the form~\eqref{psi} we then have, by Lemma~\ref{h2s2},
$$\int_{v_1}^{v_2}e^{\kappa v}\|\partial_v\zeta(v,r,\,\cdot\,)\|_{H^2(\bbS^2)}^2dv\lesssim 1\;,$$
and consequently, for such functions,
\bea
\label{psiControl1}
Z(v_2,r)\leq Ce^{-\kappa v_1} +2Z(v_1,r)\;.
\eea
Moreover, for any function of the form $\zeta=\partial_v(\partial_v^l\leo^I\phi)$ or
$\zeta=\leo_i(\partial_v^l\leo^I\phi)$,  we have, again by Lemma~\ref{h2s2},
\bea
\label{psiControl2}
\int_{v_1}^{v_2}Z(v,r)dv\leq C \int_{v_1}^{v_2}\|\zeta(v,r,\,\cdot\,)\|_{H^2(\bbS^2)}^2dv\lesssim e^{-\kappa\, v_1} \;.
\eea
Multiplying the last inequality by an arbitrary constant $L>0$ and adding the result to~\eqref{psiControl1} yields
\bea
\label{psiControl2}
Z(v_2,r)+L\int_{v_1}^{v_2}Z(v,r)dv\leq 2Z(v_1,r)+Ce^{-\kappa v_1}\;,
\eea
and the result now follows by choosing $L>2\kappa$ and applying, for each $r\in[r_0,r_+]$, Lemma~\ref{decayLemma}
to the functions $f_r(v):=Z(v,r)$.
\end{proof}

The next is an immediate consequence of the last proposition and the equations~\eqref{roundGrad}
and~\eqref{roundLap}:

\begin{cor}
\label{corDelta}
Under Assumption~\ref{energyAss}
\bea
\lb{pointDecayGrad}
\left|\nabb\psi\right|^2\lesssim e^{-\kappa\, v}\;,
\eea
and
\bea
\lb{pointDecayGrad}
\left(\Dell\psi\right)^2\lesssim e^{-\kappa\, v}\;.
\eea
\end{cor}

We now turn to the pointwise estimates of the radial derivatives.

\begin{prop}
\label{propPointDecayR}
Let $\phi$ be a solution of~\eqref{waveEq} that satisfies Assumption~\ref{energyAss}. Then, for any  $\kappa<\min\{\bar{\kappa},2\kappa_+\}$, there exists $r_0 \in (r_{-}, r_+)$,
such that,
for all $r\in[r_0,r_+]$, $\omega\in\bbS^2$ and $v\geq 0$, we have
\bea
\lb{pointDecayR}
\left(\partial_r\psi(v,r,\omega)\right)^2\lesssim e^{-\kappa\, v}\;.
\eea
\end{prop}

\begin{proof}
In our coordinates~\eqref{EFv}
the wave operator reads
\bea
\label{waveEF}
\square_g \psi= D\partial_r^2\psi +2 \partial_v\partial_r\psi
+\frac{2}{r}\partial_v\psi+\left(\frac{2}{r}D+D'\right)\partial_r\psi+\Dell\psi\;.
\eea
Therefore, for any solution $\psi$ of the wave equation we have
\bea
\label{waveNearH}
\left(\partial_v+\frac{1}{2}D\partial_r \right)\prp+G\prp= S\;,
\eea
where
$$G=\frac{1}{r}D+\frac{1}{2}D'\;,$$
and
$$S=-\frac{1}{r}\partial_v\psi-\frac{1}{2}\Dell\psi\;.$$
According to Proposition~\ref{propPointDecayVO} and Corollary~\ref{corDelta} the source term satisfies the estimate
\bea
\label{sourceEstimate}
|S|\lesssim e^{-\frac{1}{2}\kappa\,v }\;.
\eea
The characteristic of~\eqref{waveNearH}, passing through the point $(v_1,r_1,\omega_1)$ is the outgoing null line
\bea
\label{char}
\chi(v)=\chi(v;v_1,r_1,\omega_1)=(v,r(v;v_1,r_1,\omega_1),\omega_1)\;,
\eea
with $v\mapsto r(v;v_1,r_1,\omega_1)$ determined by
$$
\left\{
\begin{array}{l}
\frac{dr}{dv}=\frac{1}{2}D \\
r(v_1;v_1,r_1,\omega_1)=r_1\;.
\end{array}
\right.
$$
We note that
\bea
\label{derivativeChar}
\frac{d}{dv}\left(\prp\circ\chi(v)\right)=\left(\partial_v+\frac{1}{2}D\partial_r \right)\prp\circ\chi(v)\;.
\eea
We have $D(r)<0$, for $r_0<r<r_+$, $D(r_+)=0$ and $D'(r_+)=2\kappa_+>0$. Therefore, given $\epsilon>0$ we may choose  $r_0$ smaller,
but sufficiently close to $r_+$ such that, for $r\in[r_0,r_+]$,
\bea
\label{Gestimate}
G(r)\geq \kappa_+-\epsilon\;.
\eea
We can then fix $v_0\geq 0$ so that all characteristics with $r_1\in[r_0,r_+]$ intersect the compact set $\cC_{v_0}(r_0,r_+)$, at exactly one point, and remain in $r_0\leq r \leq r_+$, for all $v_0\leq v\leq v_1$. Using~\eqref{derivativeChar}, let us integrate~\eqref{waveNearH}
from $\cC_{v_0}(r_0,r_+)$ along the characteristics;
by taking into account~\eqref{Gestimate} and~\eqref{sourceEstimate} we obtain,
for all $v\geq v_0$, $r\in[r_0,r_+]$ and $\omega\in\bbS^2$,
\begin{eqnarray*}
\left|\prp(v,r,\omega)\right|
&\leq&
\left|\prp\circ\chi(v_0;v,r,\omega)\right|e^{-\int_{v_0}^{v}G\circ\chi(\tilde v)d\tilde v}+
\int_{v_0}^{v} \left|S\circ\chi(\tilde v;v,r,\omega)\right|e^{-\int_{\tilde v}^{v}G\circ\chi(\bar v)d\bar v}d\tilde v\;.
\\
&\leq&
Ce^{-(\kappa_+-\epsilon)(v-v_0)}+
\int_{v_0}^{v} Ce^{-\frac{1}{2}\kappa \tilde v}e^{-(\kappa_+-\epsilon)(v-\tilde v)}d\tilde v\
\\
&\lesssim&
e^{-(\min\{\frac{1}{2}\kappa,\kappa_+\}-\epsilon)v}\;.
\end{eqnarray*}
Since $\epsilon>0$ can be made arbitrarily small, we are done.

\end{proof}


\subsection{Conditional energy estimates}
\label{sectionImprov}

We will now use the conditional pointwise estimates derived in the previous section to obtain conditional energy estimates, which will later lead, by the iteration scheme of Section~\ref{secIt},  to improved energy estimates.

\begin{lem}
\label{lemDecayWr}
Let $\phi$ be a solution of~\eqref{waveEq} that satisfies Assumption~\ref{energyAss}. Then, for all $v\geq 0$ and $u\geq u_{r_0}(v)$
we have
\bea
\int\limits_ {{\cC}_{v}(u,\infty)} J_{\mu}^N[\psi] n^{\mu}_{\cC_{v}}\dVt_{\cC_{v}}\lesssim
e^{\kappa_+(v-u)-\kappa v}\;.
\eea
\end{lem}

\begin{proof}
We simply have to note that, in the region under consideration $r(u,v)\in[r_0,r_+]$ and we have by~\eqref{r-r}
\bea
\lb{exrela}
r_+-r(u,v)\lesssim_{r_0} e^{\kappa_+(v-u)}\;.
\eea
Then, in view of~\eqref{decayJ},
\begin{eqnarray*}
 \int\limits_ {{\cC}_{v}(u,\infty)} J_{\mu}^N[\psi] n^{\mu}_{\cC_{v}}\dVt_{\cC_{v}}
 &\sim&
\int_{r(u,v)}^{r_+}\int_{\bbS^2}
J_{\mu}^N[\psi] n^{\mu}_{\cC_{v}}(v,r,\omega)d\omega dr
\\
&\lesssim&
\int_{r(u,v)}^{r_+}\int_{\bbS^2}
e^{-\kappa v}d\omega dr
\\
&\lesssim&
\left(r_+-r(u,v)\right)e^{-\kappa v}\;.
\end{eqnarray*}
\end{proof}

\begin{prop}
\label{propImprovedDecay}
Let $\phi$ be a solution of~\eqref{waveEq} that satisfies Price's Law (Assumption~\ref{pricesAss}), with $2p>\kappa_+$, and Assumption~\ref{energyAss}. Then, for any  $\kappa<\min\{\bar{\kappa},2\kappa_+\}$, any $m\geq 1$ and any $\delta>0$, there exists $r_0 \in (r_{-}, r_+)$ such that, for all $r_0\leq r \leq r_+$ and $0\leq v_1 \leq v_2$,
\bea
 \int\limits_ {\Sigma_r(v_1,v_2)} J_{\mu}^N[\psi] n^{\mu}_{\Sigma_r} \dVt_{\Sigma_r}
 \lesssim e^{-\Delta_1(m)v_1}+ e^{-\Delta_2(m)v_1}\;,
\eea
for
\begin{equation}
 \label{Delta1}
 \Delta_1=\frac{1}{m}\left[(m-1)(2-\delta)\kappa_++\kappa\right]\;,
\end{equation}
 and
 \begin{equation}
 \label{Delta2}
 \Delta_2=\frac{1}{m}\left[2p-\delta+(m-1)(1-\delta)\kappa_+\right]\;.
\end{equation}
\end{prop}

\begin{proof}
Following~\cite{jan1} we consider, for $\delta>0$, the following modified red-shift vector field
\begin{equation}
\label{modRS}
\hat N_{\delta}:=e^{(1-\delta)\kappa_+ v}N\;.
\end{equation}
It turns out that~\cite{jan1}[pp. 112]
\begin{equation}
K^{\hat N_{\delta}}[\psi]\geq 0\;,
\end{equation}
provided $r_0$ is sufficiently close to $r_+$.

Then, the divergence
Theorem applied to the region $\cR$,
defined by Figure~\ref{region2}, gives rise to the following energy estimate
{\begin{figure}[ht]
\centering
\includegraphics[width=0.4\textwidth]{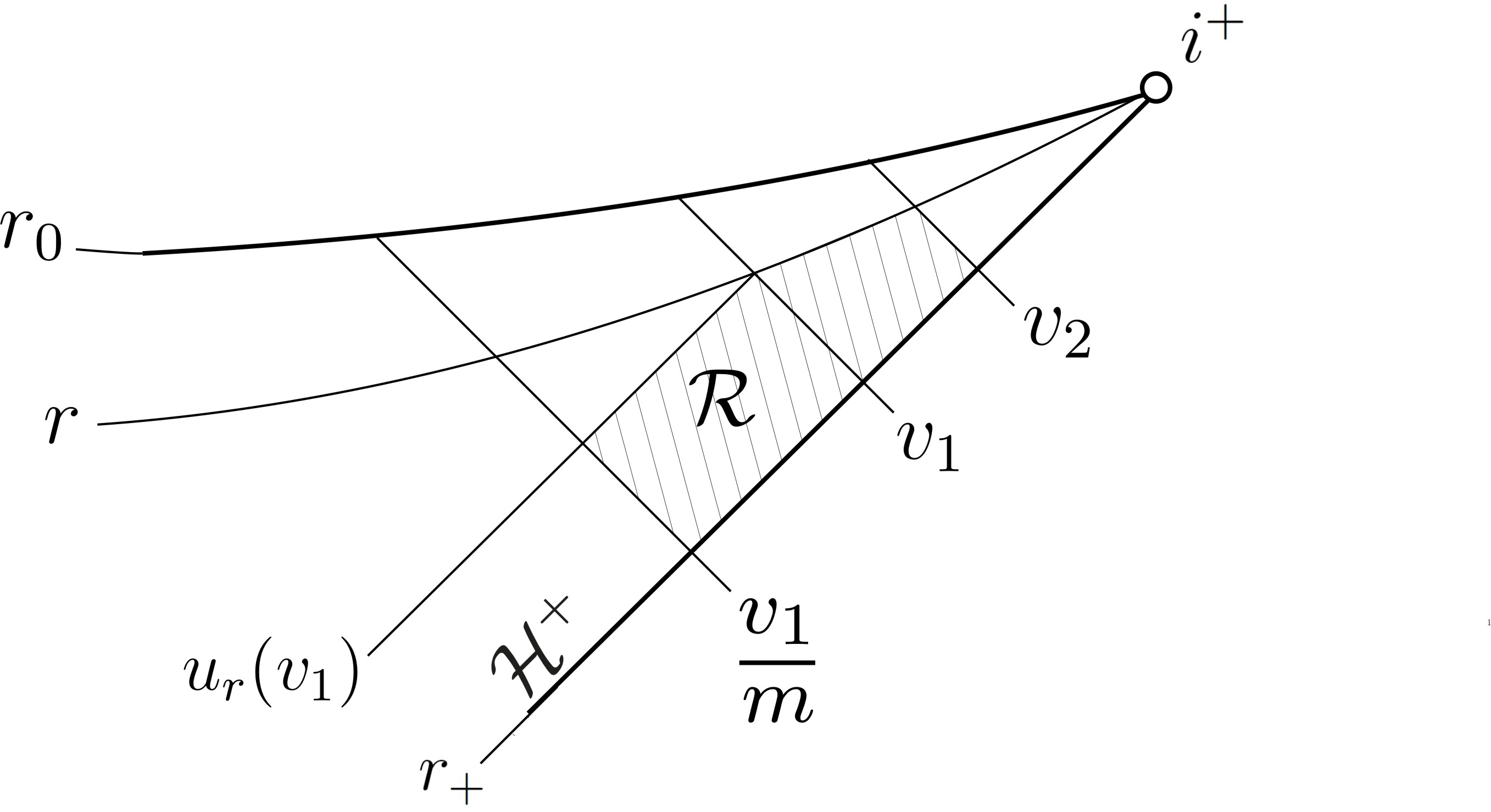}
\caption[]{Region $\cR$ represented as the hatched area.}
\label{region2}\end{figure}}
\begin{eqnarray*}
&&\int\limits_ {\Sr(v_1,v_2)} J_{\mu}^{\hat N_{\delta}}[\psi] n^{\mu}_{\Sr} \dVt_{\Sr}
\leq
\int\limits_ {{\cC}_{\frac{v_1}{m}}(u_r(v_1),\infty)} J_{\mu}^{\hat N_{\delta}}[\psi] n^{\mu}_{\cC_{v}} \dVt_{\cC_{v}} +
\int\limits_{\cH^+(\frac{v_1}{m},v_2)}J_{\mu}^{\hat N_{\delta}}[\psi]n^{\mu}_{\cH^+}\dVt_{\cH^+}\;.
\end{eqnarray*}
Now the term on the left satisfies
\begin{equation}
\int\limits_ {\Sr(v_1,v_2)} J_{\mu}^{\hat N_{\delta}}[\psi] n^{\mu}_{\Sr} \dVt_{\Sr}
\geq e^{(1-\delta)\kappa_+v_1} \int\limits_ {\Sr(v_1,v_2)} J_{\mu}^{N}[\psi] n^{\mu}_{\Sr} \dVt_{\Sr} \;.
\end{equation}
Since
$$u_r(v_1)=v_1-2r^*\;,$$
with the help of Lemma~\ref{lemDecayWr} the first term on the right can be estimated by
\begin{eqnarray*}
\int\limits_ {{\cC}_{\frac{v_1}{m}}(u_r(v_1),\infty)} J_{\mu}^{\hat N_{\delta}}[\psi] n^{\mu}_{\cC_{v}}
&=&
e^{\frac{1}{m}(1-\delta)\kappa_+ v_1}\int\limits_ {{\cC}_{\frac{v_1}{m}}(u_r(v_1),\infty)} J_{\mu}^{N}[\psi] n^{\mu}_{\cC_{v}}
\\
&\lesssim&
e^{\frac{1}{m}(1-\delta)\kappa_+ v_1}e^{\kappa_+(\frac{v_1}{m}-v_1+2r^*)}e^{-\frac{v_1}{m}\kappa}
\\
&\lesssim&
e^{2\kappa_+r^*}e^{-\frac{1}{m}\left[(m-2+\delta)\kappa_++\kappa \right]v_1}
\\
&\lesssim&
e^{-\frac{1}{m}\left[(m-2+\delta)\kappa_++\kappa \right]v_1}\;,
\end{eqnarray*}
where in the last step we used the fact that, by setting $r^*(r_0)=0$, we get
$r^*(r)\leq0$, for $r\geq r_0$.

For the final term we invoke Price's Law (Assumption~\ref{pricesAss}) and use Lemma~\ref{lemExpDecay}, for which we require the condition $2p>\kappa_+$, to obtain
\begin{eqnarray*}
\int\limits_{\cH^+(\frac{v_1}{m},v_2)}J_{\mu}^{\hat N_{\delta}}[\psi]n^{\mu}_{\cH^+}\dVt_{\cH^+}
&=&
\int\limits_{\cH^+(\frac{v_1}{m},v_2)}e^{(1-\delta)\kappa_+v}J_{\mu}^{N}[\psi]n^{\mu}_{\cH^+}\dVt_{\cH^+}
\\
&\lesssim&
e^{-\frac{1}{m}\left[2p-(1-\delta)\kappa_+-\delta \right]v_1}\;.
\end{eqnarray*}
Inserting the last three estimates into the energy estimate gives
\begin{equation*}
e^{(1-\delta)\kappa_+v_1} \int\limits_ {\Sr(v_1,v_2)} J_{\mu}^{N}[\psi] n^{\mu}_{\Sr} \dVt_{\Srn}
\lesssim
e^{-\frac{1}{m}\left[(m-2+\delta)\kappa_++\kappa \right]v_1}
+
e^{-\frac{1}{m}\left[2p-(1-\delta)\kappa_+-\delta \right]v_1}\;,
\end{equation*}
and the desired result follows by simply passing the exponential factor in the left to the right.
\end{proof}

\subsection{Iteration and improved estimates}
\label{secIt}

For the sake of simplicity let us assume for a moment that  $p=\kappa_+$ in order to illustrate the iterative mechanism that will allow us
to improve the previously established estimates; we will return to the general case shortly. By Proposition~\ref{propDecay0},
there exists an $r_0<r_+$ such that, for all $r\in[r_0,r_+]$ and any $\kappa_0<\kappa_+$,
\begin{equation}
\label{it0}
 \int_{\Sr(v_1,v_2)}J^N_{\nu}[\psi]n^{\nu}_{\Sr}\lesssim e^{-\kappa_0 v_1}\;.
\end{equation}
In other words Assumption~\ref{energyAss} and consequently all the results in Sections~\ref{sectionCond} and~\ref{sectionImprov} hold for
$\bar\kappa=\bar\kappa_0=\kappa_+$.  To improve~\eqref{it0} we can (in this simplified setting) take for instance $m=m_0=2$ and then
$\Delta_1$ and $\Delta_2$ (recall~\eqref{Delta1} and~\eqref{Delta2}) can be chosen arbitrarily close to $\frac{3}{2}\kappa_+$;
it then follows from Proposition~\ref{propImprovedDecay}, with $\bar\kappa_0=\kappa_+$, that for any
$\kappa_1<\frac{3}{2}\kappa_+$,
$$\int_{\Sr(v_1,v_2)}J^N_{\nu}[\psi]n^{\nu}_{\Sr}\lesssim e^{-\kappa_1 v_1}\;.$$
So we see that Assumption~\ref{energyAss} now holds with
$\bar\kappa=\bar\kappa_1=\frac{3}{2}\kappa_+>\bar\kappa_0$. We can then repeat the process by choosing a new $m=m_1\geq 1$, where the lower bound is needed to apply Proposition~\ref{propImprovedDecay}, that allows us to improve the decay once more. In fact, as we will see, we can keep on improving this decay, even without the previous simplifying assumption ($p=\kappa_+$), until we reach:
\begin{thm}
\label{thmRNotEst}
Let $\phi$ be a solution of~\eqref{waveEq} that satisfies Price's Law (Assumption~\ref{pricesAss}),
with $2p>\kappa_+$. Then, for any $\epsilon>0$, there exists $r_0 \in (r_{-}, r_+)$ such that, for all $r_0\leq r \leq r_+$ and $0\leq v_1\leq v_2$,
\bea
\label{optEnergy}
\int\limits_ {\Sigma_r(v_1,v_2)} J_{\mu}^N[\psi] n^{\mu}_{\Sigma_r} \dVt_{\Sigma_r}
\lesssim e^{-2(\min\{\kappa_+,p\}-\epsilon)v_1}
\eea
and, consequently,
\begin{equation}
\label{optPointWise}
\left(\partial_v\psi(v,r,\omega)\right)^2+\left(\leo_i\psi(v,r,\omega)\right)^2+\left(\partial_r\psi(v,r,\omega)\right)^2\lesssim e^{-2(\min\{\kappa_+,p\}-\epsilon)v}\;,
\end{equation}
for all functions of the form~\eqref{psi} and any $i\in\{1,2,3\}$.
\end{thm}
\begin{proof}

Note first that~\eqref{optPointWise} is an immediate consequence of~\eqref{optEnergy} together with Proposition~\ref{propPointDecayVO} and Proposition~\ref{propPointDecayR}.

From the discussion prior to the statement of this theorem we see that the main result of interest~\eqref{optEnergy} follows from  the following:
\begin{lem}
 For any $\epsilon>0$, there exists a $\delta>0$ and a sequence  $\{m_n\}_{n\in\bbN_0}$, satisfying $m_n\geq1$, $\forall n$,
 such that the sequence $\{\kappa_n\}_{n\in\bbN_0}$ defined by
$$
\left\{
\begin{array}{l}
\kappa_0<\kappa_+ \\
 \kappa_{n+1}=
 \Delta_1(m_n,\kappa_n) =\frac{1}{m_n}\left[(m_n-1)(2-\delta)\kappa_++\kappa_n\right]\;,
\end{array}
\right.
$$
satisfies the following:
\begin{enumerate}
\item  $\kappa_{n+1}=\Delta_2(m_n)$, $\forall n$, and
\item there exists $N\in\bbN$ for which $\kappa_N>2\min\{\kappa_+,p\}-\epsilon$\;.
\end{enumerate}
\end{lem}
\begin{proof}

Fix $\epsilon>0$ and choose $m_n$ as the solution of
$$\Delta_1(m_n,\kappa_n)=\Delta_2(m_n)\;,$$
i.e.,
\begin{equation}
\label{mDef}
m_n=\frac{2p-\delta-\kappa_n+\kappa_+}{\kappa_+}\;.
\end{equation}
Then by construction $\kappa_{n+1}=\Delta_2(m_n)$,  $\forall n$.
Observe also that if
\begin{equation}
\label{why2p>k}
2p>\kappa_n\;,
\end{equation}
we can choose $\delta$ sufficiently small in order to make sure that the condition $m_n\geq 1$ holds.

For the choice~\eqref{mDef} the recursive rule for our sequence becomes
\begin{equation}
\label{knDef}
\kappa_{n+1}
=\frac{\kappa_+}{2p-\delta+\kappa_+-\kappa_n}\left[(2-\delta)(2p-\delta)-(1-\delta)\kappa_n\right] \;.
\end{equation}

If it happens that there exists $N\in\bbN$ such that $\kappa_N\geq 2\min\{\kappa_+,p\}$ then we are done.

So we are left with analyzing the situation when
\begin{equation*}
 \kappa_n<2\min\{\kappa_+,p\}\;,\;\forall n\;.
\end{equation*}
We consider two cases:

{\em Case 1:} $p\geq\kappa_+$.

In this case, our latest assumption gives
$$\kappa_n<2\kappa_+\;,\;\forall n\;.$$
Then the denominator of~\eqref{knDef} reads
$$2p-\delta+\kappa_+-\kappa_n>2p-\delta-\kappa_+$$
which can be made positive by choosing $\delta$ small enough. Consequently
\begin{eqnarray*}
 \kappa_{n+1}-\kappa_n\geq 0
 &\Leftrightarrow&
 \kappa_+\left[(2-\delta)(2p-\delta)-(1-\delta)\kappa_n\right]
-\kappa_n\left[2p-\delta+\kappa_+-\kappa_n\right]\geq 0\;.
\end{eqnarray*}
But we have
\begin{eqnarray*}
&& \kappa_+\left[(2-\delta)(2p-\delta)-(1-\delta)\kappa_n\right]
-\kappa_n\left[2p-\delta+\kappa_+-\kappa_n\right]=0
\\
&\Leftrightarrow&
\kappa_n=\kappa_{\pm}^{\delta}
:=\frac{1}{2}\big[(2-\delta)\kappa_++2p-\delta\pm
\left|(2-\delta)\kappa_+ - 2p-\delta\right|\big]\;.
\end{eqnarray*}
Since we are analyzing the case $p\geq\kappa_+$ we get
$$\kappa_-^{\delta}=(2-\delta)\kappa_+ \quad\text{and}\quad \kappa_+^{\delta}=2p-\delta.$$
We will now proceed by contradiction by assuming  that for all $\delta>0$ we have (for the $\epsilon>0$ fixed in beginning of the proof)
$$\kappa_n<2\kappa_+-\epsilon\;,\;\forall n\;.$$
Then, by making $\delta$ small we can make sure that
$2\kappa_+-\epsilon< \kappa_-^{\delta}$. But then $\kappa_n$ is increasing and since it is bounded it is therefore convergent;
in fact the previous computations show that it converges to $\kappa_-^{\delta}$ which gives the desired contradiction.

{\em Case 2:} $p<\kappa_+$.

The analysis of this case is similar to the previous with $\kappa_-^{\delta}=2p-\delta$.
\end{proof}
This concludes the proof of Theorem~\ref{thmRNotEst}.
\end{proof}

\subsection{Propagating the energy estimates to the Cauchy horizon}
\label{secProp}

We can now follow Sbierski~\cite{jan1} all the way to the Cauchy horizon. For the sake of completeness we present a brief sketch of how
the energy estimates propagate; the relevant details can be found in~\cite{jan1}[pp. 113--118].

First one uses the fact that
$$\int\limits_ {\Sr(u,\infty)} J_{\mu}^{N_1}[\psi] n^{\mu}_{\Sigma_r}\sim \int\limits_ {\Sr(u,\infty)} J_{\mu}^{N_2}[\psi] n^{\mu}_{\Sigma_r}\;,$$
for any pair, $N_1$, $N_2$, of timelike and future directed vector fields that commute with the stationary Killing vector field $T$.
Moreover for any such vectors, if $r_-<r_1<r_0<r_+$, we have the estimate
\begin{eqnarray*}
\int\limits_ {\Sigma_{r_1}(u,\infty)} J_{\mu}^{N_i}[\psi] n^{\mu}_{\Sigma_r}
&\leq&
C(r_0,r_1)\int\limits_ {\Srn(u,\infty)} J_{\mu}^{N_i}[\psi] n^{\mu}_{\Sigma_r}
\\
&\lesssim&
e^{-2(\min\{\kappa_+,p\}-\epsilon)v_{r_0}(u)}
\\
&\lesssim&
e^{-2(\min\{\kappa_+,p\}-\epsilon)u}\;.
\end{eqnarray*}
Where we have used the fact that $r^*(r_0)=0$.
One then considers, for $\delta>0$, the vector field
\begin{equation}
 \label{modBS}
\check N_{\delta}=e^{(1+\delta)\kappa_- u}\check N\;,
\end{equation}
where $\check N$ is the blue-shift vector field, which is timelike, future directed and such that $[\check N,T]=0$.
It turns out that
$$K^{\check N_{\delta}}[\psi]\geq 0\;,$$
in $r_-\leq r\leq r_1$, for some $r_1$ sufficiently close to $r_-$.

Stokes' theorem applied to the divergence of $\check N_{\delta}$, in the region defined by Figure~\ref{final},
{\begin{figure}[ht]
\centering
\includegraphics[width=0.4\textwidth]{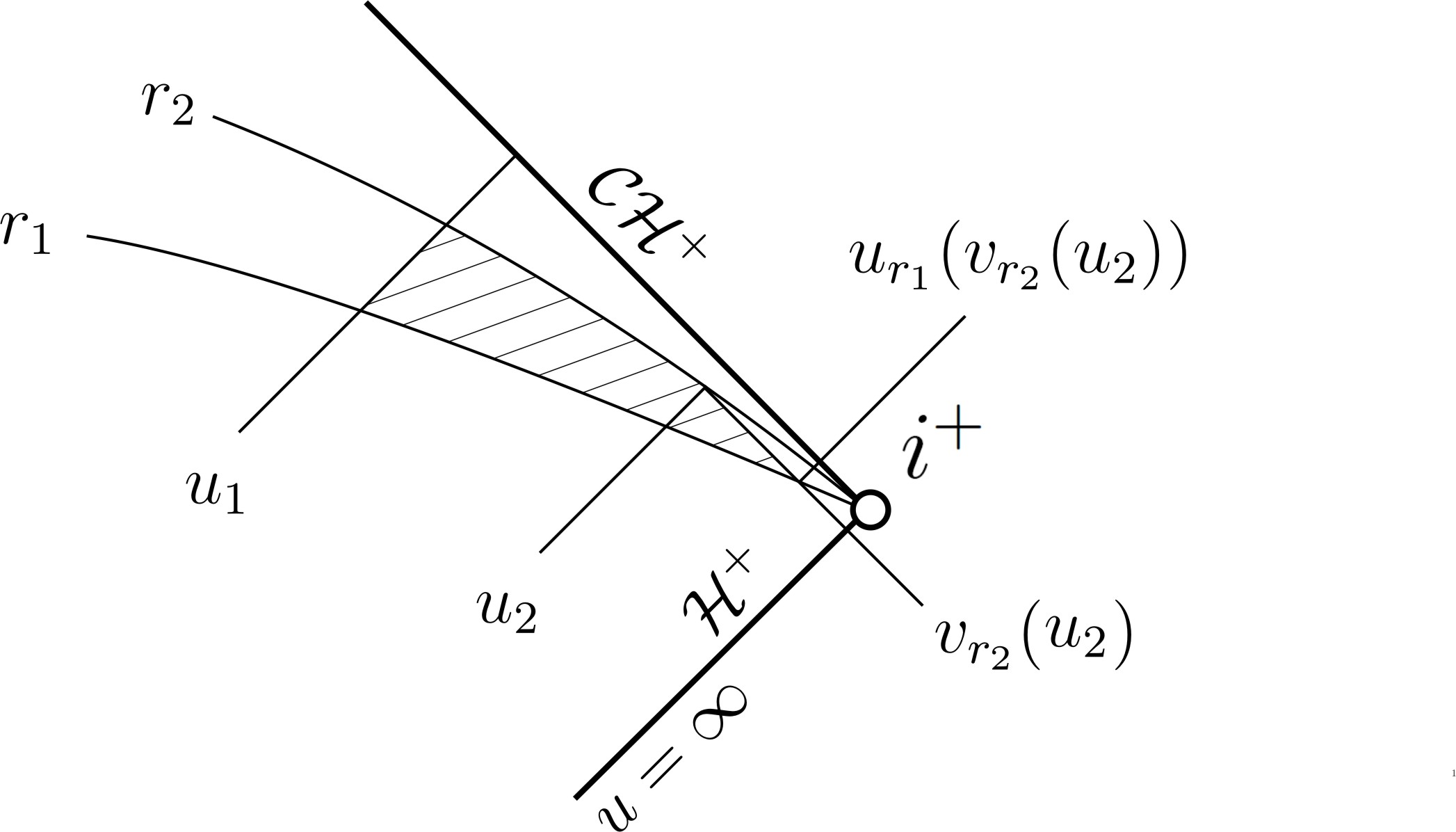}
\caption[]{Region in which we carry out divergence theorem is represented as the hatched area.}
\label{final}\end{figure}}
leads to
\bea
\label{lastEnergyEst}
\int\limits_ {{\cal C}_{u_1}(r_2,r_1)} e^{(1+\delta)\kappa_- u} J_{\mu}^{\check N}[\psi] n^{\mu}_{{\cal C}_u}
+ \int\limits_{\Sigma_{r_2}(u_1,u_2)} e^{(1+\delta)\kappa_- u} J_{\mu}^{\check N}[\psi] n^{\mu}_{\Sigma_r}
&\leq& \int\limits_{\Sigma_{r_1}(u_1,u_{r_1}(v_{r_2}(u_2)))} e^{(1+\delta)\kappa_- u} J_{\mu}^{\check N}[\psi] n^{\mu}_{\Sigma_r}\nonumber \\
&\leq& \int\limits_{\Sigma_{r_1}(u_1,\infty)} e^{(1+\delta)\kappa_- u} J_{\mu}^{\check N}[\psi] n^{\mu}_{\Sigma_r}\;.
\eea
Now, in view of the previous discussion, Theorem~\ref{thmRNotEst} and Lemma~\ref{lemExpDecay}, we see that provided that
\begin{equation}
2\min\{p,\kappa_+\}>\kappa_-\;,
\end{equation}
for any $\delta>0$ sufficiently small we have
\begin{equation}
\label{finalDecay1}
\int\limits_{\Sigma_{r_1}(u_1,\infty)} e^{(1+\delta)\kappa_- u} J_{\mu}^{\check N}[\psi] n^{\mu}_{\Sigma_r}\lesssim e^{-(2\min\{p,\kappa_+\}-\kappa_--\delta)u_1}\;.
\end{equation}
Inserting the last estimate into~\eqref{lastEnergyEst} gives
\begin{equation}
\label{finalDecay2}
\int\limits_ {{\cal C}_{u_1}(r_2,r_1)}  J_{\mu}^{\check N}[\psi] n^{\mu}_{{\cal C}_u}
+ \int\limits_{\Sigma_{r_2}(u_1,u_2)}  J_{\mu}^{\check N}[\psi] n^{\mu}_{\Sigma_r}
 \lesssim e^{-(2\min\{p,\kappa_+\}-\delta)u_1}\;,
\end{equation}
which gives~\eqref{mainEst}. As a consequence we see that $\phi\in H_{loc}^1(\Mint\cup\cH^+\cup{\cal CH}^+)$.

Uniform boundedness~\eqref{mainBound} then follows by a simple adaptation of the initial part of the proof of Proposition~\ref{propPointDecayVO}
(compare with~\cite{anne}[Section 14]) and the continuity statement, in Theorem~\ref{mainThm}, can then be proven by using the ideas of~\cite{anne}[Section 15].

\section*{Acknowledgement}

We would like to thank: Jan Sbierski for sharing is PhD thesis and for useful discussions; Pedro Gir\~ao, Jos\'e Nat\'ario and Jorge Drumond Silva for various comments concerning a preliminary version of this work; Jaques Smulevici for taking the time to explain to us a version of Lemma~\ref{decayLemma}; Peter Hintz for various clarifications concerning his work.

Both authors thank the Institute Henri Poncair\'e, for the hospitality and partial financial support during the program ``Mathematical General Relativity, were this research was started.

This work was partially supported by FCT/Portugal through UID/MAT/04459/2013 and grant (GPSEinstein) PTDC/MAT-ANA/1275/2014.



\end{document}